\documentclass[11pt,leqno]{article}
\usepackage{amsmath,amssymb,amsthm,amsfonts}
\usepackage{epsfig,graphicx}
\usepackage{verbatim,url,wasysym,marvosym,relsize,color}
\usepackage[body={6in,8.5in}]{geometry}

\theoremstyle{plain}
\begingroup
\newtheorem{theorem}{Theorem}[section]

\newtheorem{lemma}[theorem]{Lemma}

\endgroup
\theoremstyle{definition}

\theoremstyle{remark}

\numberwithin{equation}{section}
\numberwithin{figure}{section}

\newcommand{\nwc}{\newcommand}
\nwc{\grad}{\nabla}
\nwc{\pd}{\partial}
\nwc{\RR}{\mathbb{R}}
\nwc{\NN}{\mathbb{N}}
\nwc{\skorohod}{{\mathbb{D}(\mathbb{R})}}
\nwc{\spacetime}{{\mathbb{R}\times[0,\infty)}}
\nwc{\cross}{\times}
\nwc{\Id}{\mbox{Id}}
\nwc{\loc}{\mbox{loc}}
\nwc{\ninN}{n\in\naturals}
\nwc{\indicator}[1]{\mathbb{1}_{\left[ {#1} \right] }}
\nwc{\ind}{1\hspace{-2.4mm}{1}}

\nwc{\Lboth}{\cal{B}}
\nwc{\Uboth}{\cal{C}}
\nwc{\Linit}{A}
\nwc{\Uinit}{D}
\nwc{\Lstar}{A^*}
\nwc{\Ustar}{D^*}

\nwc{\lboth}{\cal{A}}
\nwc{\uboth}{\cal{D}}
\nwc{\linit}{B'}
\nwc{\uinit}{C'}
\nwc{\lstar}{{B^*}'}
\nwc{\ustar}{{C^*}'}

\begin{document}
\title{Rates of convergence for Smoluchowski's coagulation equations}
\author{Ravi Srinivasan\textsuperscript{1}}
\date{\today}

\maketitle
\begin{abstract}
We establish nearly optimal rates of convergence to self-similar solutions of
Smoluchowski's coagulation equation with kernels $K = 2$, $x + y$, and $xy$.
The method is a simple analogue of the Berry-Ess\'een theorem in classical
probability and requires minimal assumptions on the initial data, namely that
of an extra finite moment condition. For each kernel it is shown that the
convergence rate is achieved in the case of monodisperse initial data.
\end{abstract}
\footnotetext[1]
{Department of Mathematics, The University of Texas at Austin, Austin, TX 78712,
Email: \tt{rav@math.utexas.edu}}
%
%
%
%
%
%
\section{Introduction}
\label{sec:intro}

Smoluchowski's coagulation equation
\begin{equation}
  \label{eqn:SmoluchowskiContinuous} \partial_t n ( t, x ) = \frac{1}{2}
  \int_0^x K ( x - y, y ) n ( t, x - y ) n ( t, y ) d y - \int_0^{\infty} K (
  x, y ) n ( t, x ) n ( t, y ) d y
\end{equation}
is a fundamental mean-field model for cluster growth that arises in a wide
range of fields, including physical chemistry, astrophysics, and the dynamics
of biological systems (see \cite{Aldous} for a review). Here,
$n ( t, x )$ is the density of the number distribution of clusters of
size $x \in (0,\infty)$ at time $t \geq 0$ and $K ( x, y )$ is a symmetric rate
kernel. In the case of discrete sizes $l \in \{1, 2, \cdots \}$ an analogous equation
for the coefficients $n_l(t)$ of the number distribution is
\begin{equation}
  \label{eqn:SmoluchowskiDiscrete} \partial_t n_l( t ) = \frac{1}{2} \sum_{j
  = 1}^{l - 1} \kappa_{ l - j, j } n_{l - j}( t ) n_j( t ) - \sum_{j =
  1}^{\infty} \kappa_{ l, j } n_l( t ) n_j( t ) ,
\end{equation}
with $\kappa_{l,j} = K( l , j )$. The continuous and discrete cases can be
considered together via the weak formulation of Smoluchowski's coagulation equation,
given in terms of a moment identity for the number distribution $n ( t, d x )$:
\begin{equation}
  \label{eqn:Smoluchowski} \partial_t \int_{( 0, \infty )} \phi ( x ) n ( t, d x )
  = \frac{1}{2} \int_{( 0, \infty )} \int_{( 0, \infty )} ( \phi ( x + y ) - \phi ( y )
  - \phi ( x ) ) K ( x, y ) n ( t, d y ) n ( t, d x ),
\end{equation}
where $\phi$ is a suitable test function. We direct the interested reader to
{\cite{MP1,MP2}} for further discussion on the well-posedness and asymptotic behavior
of measure-valued solutions $n ( t, d x )$ to (\ref{eqn:Smoluchowski}). A general survey
of existing literature in stochastic coalescence is given in {\cite{Aldous}}, with more
recent work reviewed in {\cite{Pego}}. Note that throughout this paper we try to keep
the same notation as in {\cite{MP2}}.

The present work is restricted to the homogeneous, `solvable' kernels $K = 2,
x + y, x y$, for which $K ( \alpha x, \alpha y ) = \alpha^{\gamma} K (
x, y )$ with $\gamma = 0, 1, 2$. It has been shown in {\cite{MP1}}
that (\ref{eqn:Smoluchowski}) admits a one-parameter family of self-similar
solutions whose domains of attraction under dynamic scaling are characterized by the
tails of the initial data. In particular, there is a unique self-similar solution with
finite $( \gamma + 1 )$th moment, which has exponentially decaying tails and
attracts all initial data that satisfy this finite moment condition. In
{\cite{MP2}}, it was shown that with an additional integrability
hypothesis on the Fourier transform of the initial data one has uniform
convergence of densities (in the continuous case) or coefficients (in the
discrete case) to a self-similar solution as $t$ approaches the time horizon
$T_{\gamma}$, where $T_{\gamma} = \infty$ for $\gamma = 0, 1$ and $T_{\gamma}
= 1$ for $\gamma = 2$. The proof of this result is similar to that of uniform
convergence of densities in the central limit theorem (see Feller
{\cite{Feller}}, Section XV.5).

Here we establish near-optimal $L^{\infty}$-rates of convergence to the exponentially
decaying self-similar solutions of Smoluchowski's coagulation equation with $K
= 2, x + y, x y$. Continuing the analogy with the classical central limit theorem (CLT),
this result corresponds to the Berry-Ess\'een theorem for rates of convergence to the normal
law ({\cite{Feller}}, XVI.5). The method is simple and robust: it holds
for all initial distributions, requiring only an additional finite moment
condition beyond those needed for well-posedness and convergence to the self-similar profile. In
particular, it is true when the initial distribution has a density or is a
lattice measure. Our work improves upon recent results of Ca{\~n}izo, Mischler, and Mouhot
{\cite{CMM}} as it holds for a large class of initial data and for all of the solvable kernels.
More broadly, it once again demonstrates the utility of applying methods from
classical probability to study asymptotic behavior under rescaling for integral equations
of convolution type. Additional estimates that make use of the present work,
such as large deviations theorems, will be developed elsewhere.

We present the results for the solvable kernels in a unified framework. With $t_0 = 1$ for
$\gamma = 0$, $t_0 = 0$ for $\gamma = 1, 2$, and $n_0 ( d x ) = n ( t_0, d x
)$, define the moments $\mu_j = \int_{( 0, \infty )} x^j n_0 ( d x )$ of the
initial data. Assume that the $\gamma$th and $( \gamma + 1 )$th moments are
finite, for which we scale $x$ and $n_0$ so that $\mu_{\gamma} = \mu_{\gamma +
1} = 1$. As shown in {\cite{MP2}}, finiteness of the $\gamma$th moment ensures well-posedness of
(\ref{eqn:Smoluchowski}), while finiteness of the $( \gamma + 1 )$th
moment guarantees that the initial data are in the domain of attraction of the
self-similar solution with exponential decaying tails. These self-similar solutions are
absolutely continuous and are explicitly given in terms of their densities (see {\cite{Aldous,MP1,MP2}})
by
\begin{equation}
  \label{eqn:LimitingDist} n ( t, x ) = \frac{m_{\gamma} ( t
  )}{\lambda_{\gamma} ( t )^{\gamma + 1}} \hat{n}_{*, \gamma} \left(
  \frac{x}{\lambda_{\gamma} ( t )} \right).
\end{equation}
Here, the profiles $\hat{n}_{*, \gamma}$ for $\hat{x} \geq 0$ take the form
\begin{equation}
  \label{eqn:SelfSimilarProfiles} \hat{n}_{*, 0} ( \hat{x} ) = e^{-
  \hat{x}}, \qquad \hat{x} \hat{n}_{*, 1} ( \hat{x} ) = \hat{x}^2
  \hat{n}_{*, 2} ( \hat{x} ) = \frac{1}{\sqrt{2 \pi}} \hat{x}^{- 1 / 2}
  e^{- \hat{x} / 2}
\end{equation}
and the time-dependent moments $m_{\gamma}$ and scalings $\lambda_{\gamma}$
are
\begin{equation}
  \label{eqn:Moments} m_0 ( t ) = t^{- 1}, \qquad m_1 ( t ) = 1,
  \qquad m_2 ( t ) = ( 1 - t )^{- 1}
\end{equation}
\begin{equation}
  \label{eqn:Scalings} \lambda_0 ( t ) = t, \qquad \lambda_1 ( t ) =
  e^{2 t}, \qquad \lambda_2 ( t ) = ( 1 - t )^{- 2} .
\end{equation}

The main result of our work is as follows. Define the time parameter
$\tau_{\gamma} ( t ) = \int_{t_0}^t m_{\gamma} ( s ) ds$,
explicitly given by
\begin{equation}
  \label{eqn:TimeParameter} \tau_0 ( t ) = \log t , \qquad \tau_1 ( t
  ) = t, \qquad \tau_2 ( t ) = \log ( 1 - t )^{- 1}.
\end{equation}
If we assume finiteness of the $( \gamma + 2 )$th moment of the initial data,
we have exponentially fast convergence to self-similar form:

\begin{theorem}
  \label{thm:Main}Let $n_0$ be a positive measure with $\int_{( 0, \infty )}
  x^{\gamma} n_0 ( d x ) = \int_{( 0, \infty )} x^{\gamma + 1} n_0 ( d x ) =
  1$. Suppose that the additional finite moment assumption $\int_{( 0, \infty )}
  x^{\gamma + 2} n_0 ( d x ) < \infty$ holds and let $n ( t, d x )$ be the
  measure-valued solution to Smoluchowski's equation with $K = 2$, $x + y$, or
  $x y$ and initial data $n_0 ( dx )$. With rescaled solution
  \begin{equation}
    \label{eqn:RescaledSoln} \hat{n} ( \tau_{\gamma}, d \hat{x} ) : =
    \frac{\lambda_{\gamma} ( t )^{\gamma}}{m_{\gamma} ( t )} n ( t,
    \lambda_{\gamma} ( t ) d \hat{x} )
  \end{equation}
  define the corresponding distribution functions
  \begin{equation}
    \label{eqn:cdf} F_{\gamma} ( \tau_{\gamma}, \hat{x} ) = \int_{( 0, \hat{x}
    ]} \hat{y}^{\gamma} \hat{n} ( \tau_{\gamma}, d \hat{y} ), \qquad
    F_{*, \gamma} ( \hat{x} ) = \int_{( 0, \hat{x} ]} \hat{y}^{\gamma}
    \hat{n}_{*, \gamma} ( \hat{y} ) d\hat{y}.
  \end{equation}
  Then for $\tau_{\gamma} \in [ 0, \infty )$,
  \[ \sup_{\hat{x} > 0} |F_{\gamma} ( \tau_{\gamma}, \hat{x} ) - F_{*,
     \gamma} ( \hat{x} ) | \leq C ( \mu_{\gamma + 2} ) ( 1 + \tau_{\gamma} )
     e^{- \tau_{\gamma}} \]
  where $C ( \mu_{\gamma + 2} )$ is a constant that depends only on
  $\mu_{\gamma + 2}$.
\end{theorem}

Since we are primarily concerned with an asymptotic rate of decay to
self-similarity, we make no effort to optimize the constant in front of the
exponential term. We also demonstrate that exponential decay rates are
achieved in the case of monodisperse initial data $n_0 ( d x ) = \delta_1 ( d
x )$. Specifically, we show that for these initial conditions
\[ \sup_{\hat{x} > 0} |F_{\gamma} ( \tau_{\gamma}, \hat{x} ) - F_{*,
   \gamma} ( \hat{x} ) | = O ( e^{- \tau_{\gamma}} ) . \]

The method of proof for these results requires us to work in the Fourier
domain. To do so we make use of a smoothing argument given by Feller {\cite{Feller}},
XVI.3, which for completeness we restate here without proof. The distribution
functions in (\ref{eqn:cdf}) satisfy the assumptions of the lemma.

\begin{lemma}
  \label{lem:Smoothing}Let $F$ and $F_*$ be probability distribution
  functions and assume $|F_*' ( x ) | \leq 1$. Consider the Fourier
  transforms of their corresponding measures,
  \[ u ( i k ) = \int_{( - \infty, \infty )} e^{- i k x} F ( d x ),
     \qquad u_* ( i k ) = \int_{( - \infty, \infty )} e^{- i k x}
     F_* ( d x ) . \]
  Assume that $u' ( 0 ) = u_*' ( 0 )$ ---that is, equality of the
  first moment of $F ( d x )$ and $F_* ( d x )$. With the mollifier
  \[ \psi_T ( x ) = \frac{1 - \cos ( T x )}{\pi T x^2} \]
  define
  \[ \Delta ( x ) = F ( x ) - F_* ( x ), \qquad \Delta_T ( x ) =
     \psi_T ( x ) * \Delta . \]
  Then
  \begin{equation}
    \label{eqn:Smoothing1} \sup_x | \Delta ( x ) | \leq 2 \sup_x | \Delta_T (
    x ) | + \frac{24}{\pi T} .
  \end{equation}
\end{lemma}
In terms of Fourier transforms, (\ref{eqn:Smoothing1}) is
\begin{equation}
  \label{eqn:Smoothing2} \sup_x |F ( x ) - F_* ( x ) | \leq \frac{1}{\pi}
  \sup_x \left| \int_{- T}^T \frac{e^{i k x}}{i k} ( u ( i k ) - u_* ( i k
  ) ) d k \right| + \frac{24}{\pi T} .
\end{equation}

In the following sections, we will need to approximate Fourier transforms near
the origin by their Taylor expansions. To do so, we will use the following basic result
also due to Feller (see {\cite{Feller}}, XV.4 and the subsequent appendix).
Let $\mathbb{C}_+ = \{z \in \mathbb{C} : \text{Re} ( z ) > 0 \}$
and $\mathbb{C}_- = \{z \in \mathbb{C} : \text{Re} ( z ) < 0 \}$ be the open
left- and right-half of the complex plane, and denote by $\bar{\mathbb{C}}_+$
and $\bar{\mathbb{C}}_-$ their respective closures. Let $F(dx)$ be a
probability measure supported on $[0,\infty)$.  $F(dx)$ is then characterized
by its Fourier-Laplace transform
\begin{equation}
  \label{eqn:FourierLaplace} u(s) = \int_{[0,\infty)} e^{-sx} F(dx), \hspace{1em}
  s \in \bar{\mathbb{C}}_+ ,
\end{equation}
which is approximated as follows:

\begin{lemma}
  \label{lem:TaylorApproximation}
  Suppose $F(dx)$ is a probability measure on $[0,\infty)$ with finite $\gamma$th moment
  \begin{equation}
	\label{eqn:TaylorMoments}
	M_{\gamma} := \int_{[0,\infty)} x^{\gamma} F(dx) < \infty ,
  \end{equation}
  where $\gamma \geq 1$ is an integer. Then for $\xi, s \in \bar{\mathbb{C}}_+$,
  the Fourier-Laplace transform (\ref{eqn:FourierLaplace}) satisfies
  \begin{equation}
	\label{eqn:TaylorApproximation}
	\left| u(\xi + s) - \left( u(\xi) + \frac{u'(\xi)}{1!} s + \cdots
	+ \frac{u^{(\gamma-1)}(\xi)}{(\gamma-1)!} s^{\gamma-1} \right) \right|
	\leq \frac{M_\gamma}{\gamma!} |s|^{\gamma} .
  \end{equation}
\end{lemma}

\begin{proof}
  We begin by showing that for $s \in \bar{\mathbb{C}}_-$,
\begin{equation}
  \label{eqn:TaylorExponential} \left| e^s - \left(1 + \frac{s}{1!} + \cdots
  + \frac{s^{\gamma-1}}{(\gamma-1)!} \right) \right| \leq \frac{|s|^\gamma}{\gamma!}.
\end{equation}
  Define $g_1(s) = \int_0^s e^z dz$ and $g_{\gamma}(s)
  = \int_0^s g_{\gamma-1}(z) dz$, where the integrals are taken over
  the straight line segment beginning at the origin and ending and $s$.
  A simple calculation shows that $g_{\gamma}(s) = e^s - \left(1 + \frac{s}{1!}
  + \cdots + \frac{s^{\gamma-1}}{(\gamma-1)!} \right)$. Then (\ref{eqn:TaylorExponential})
  follows by induction. To see this, note that
  \[ |g_1(s)| = \left| s \int_0^1 e^{\tau s} d\tau \right| \leq |s|, \hspace{1em}
  s \in \bar{\mathbb{C}}_- \]
  since $|e^{\tau s}| \leq 1$. Assuming $|g_{\gamma-1}(s)| \leq
  \frac{|s|^{\gamma-1}}{(\gamma-1)!}$ we obtain the desired estimate:
  \[ |g_\gamma(s)| = \left| s \int_0^1 g_{\gamma-1}(\tau s) d\tau \right|
  \leq \frac{|s|^{\gamma}}{(\gamma-1)!} \int_0^1 \tau^{\gamma-1} d\tau
  = \frac{|s|^{\gamma}}{{\gamma}!}. \]
  Equation (\ref{eqn:TaylorApproximation}) now follows directly. Using
  (\ref{eqn:TaylorExponential}) and the definition (\ref{eqn:TaylorMoments}) yields
  \begin{align}
  & \left| u(\xi + s) - \left( u(\xi) + \frac{u'(\xi)}{1!} s + \cdots
	+ \frac{u^{(\gamma-1)}(\xi)}{(\gamma-1)!} s^{\gamma-1} \right) \right| \nonumber\\
  & = \left| \int_{[0,\infty)} e^{-\xi x} \left\{ e^{-sx} - \left( 1 - \frac{sx}{1!} + \cdots
    + (-1)^{\gamma-1}\frac{(sx)^{\gamma-1}}{(\gamma-1)!} \right) \right\} F(dx) \right|
  \leq \frac{M_\gamma}{\gamma!} |s|^{\gamma} . \nonumber
  \end{align}
\end{proof}
Taking $\xi = 0$ in (\ref{eqn:TaylorApproximation}) and using that $M_\gamma = (-1)^{\gamma} u^{(\gamma)}(0)$,
gives the particularly useful expression
\begin{equation}
	\label{eqn:TaylorApproximationZero}
	\left| u(s) - \left( 1 - \frac{M_1}{1!} s + \cdots
	+ (-1)^{\gamma-1} \frac{M_{\gamma-1}}{(\gamma-1)!} s^{\gamma-1} \right) \right|
	\leq \frac{M_\gamma}{\gamma!} |s|^{\gamma}, \hspace{1em} s \in \bar{\mathbb{C}}_+.
\end{equation}

Let us outline the remainder of the paper. We first prove Theorem \ref{thm:Main} in the
case of constant coagulation kernel in Section \ref{sec:K=2}, where the method of
proof is simpler and more easily applied. The same method is then used in Section
\ref{sec:K=x+y} to establish the result for the additive kernel. Finally, the
result for the multiplicative kernel is given in Section \ref{sec:K=xy} by transforming
to the case of additive kernel through a well-known change of variables.
\section{Rate of convergence for the constant kernel $K = 2$}
\label{sec:K=2}

In this section we restate and prove Theorem \ref{thm:Main} for the constant
kernel case. That is, we show:

\begin{theorem}
  \label{thm:K=2}Let $n_0$ be a positive measure such that $\int_{( 0, \infty )} n_0 (
  d x ) = \int_{( 0, \infty )} x n_0 ( d x ) = 1$, and $\mu_2 = \int_{( 0,
  \infty )} x^2 n_0 ( d x ) < \infty$. With $\tau = \log t$ and in terms of
  the rescaling (\ref{eqn:Moments}-\ref{eqn:Scalings}), let $\hat{n} ( \tau, d
  \hat{x} )$ be the rescaled solution (\ref{eqn:RescaledSoln}) to
  Smoluchowski's equation with $K = 2$ and initial data $n_0 ( d x )$. Then
  with
  \[ F ( \tau, \hat{x} ) = \int_{( 0, \hat{x} ]} \hat{n} ( \tau, d \hat{y} ),
     \qquad F_* ( \hat{x} ) = \int_{( 0, \hat{x} ]} \hat{n}_{*,
     0} ( \hat{y} ) d\hat{y} . \]
  and $\tau \in [ 0, \infty )$,
  \begin{equation}
    \label{eqn:K=2Result} \sup_{\hat{x} > 0} |F ( \tau, \hat{x} ) - F_* (
    \hat{x} ) | \leq C ( \mu_2 ) ( 1 + \tau ) e^{- \tau}
  \end{equation}
  where $C ( \mu_2 )$ is a constant that depends only on $\mu_2$.
\end{theorem}

\subsection{Proof for the constant kernel} \label{subsec:K=2,Proof}

The proof is can be summarized as follows. For $s \in
\bar{\mathbb{C}}_+$ define the Fourier-Laplace transforms of the rescaled
number measures as
\[ u ( \tau, s ) = \int_{( 0, \infty )} e^{- s \hat{x}} \hat{n} ( \tau, d
   \hat{x} ), \qquad u_* ( s ) = \int_{( 0, \infty )} e^{- s
   \hat{x}} \hat{n}_* ( \hat{x} ) d \hat{x} \]
and let $u_0 ( s ) = u ( 0, s )$. By (\ref{eqn:Smoothing2}), we have that
\begin{equation}
  \label{eqn:K=2Smoothing1} \sup_{\hat{x} > 0} |F ( \tau, \hat{x} ) - F_*
  ( \hat{x} ) | \leq \frac{1}{\pi} \sup_{\hat{x} > 0} \left| \int_{- iT}^{iT}
  \frac{e^{\sigma \hat{x}}}{\sigma} ( u ( \tau, \sigma ) - u_* ( \sigma )
  ) d \sigma \right| + \frac{24}{\pi T}
\end{equation}
The uniform rate of convergence of the distribution function is thus given by
the convergence rate of $u ( \tau, s )$ to $u_* ( s )$. As shown in
{\cite{MP1,MP2}}, $u$ satisfies a first-order PDE that
can be solved in terms of $u_0$ by the method of characteristics in
$\bar{\mathbb{C}}_+$. The finiteness condition on $\mu_2$ is a statement
about the decay of the tail of $\hat{n}_0$, and hence, about the regularity
of $u_0$ near the origin. Using (\ref{eqn:TaylorApproximationZero}) to estimate the
difference $u_0 ( s ) - u_* ( s )$ near $s = 0$, we can propagate this
estimate outward at the growth rate of the characteristics.
Combining this with Lemma \ref{lem:Smoothing} gives us the desired result.

{\emph{1. Evolution of characteristics:}} As shown in {\cite{MP1}},
the Fourier-Laplace transform $u ( \tau, s )$ solves the PDE
\begin{equation}
  \label{eqn:K=2PDE} \partial_{\tau} u + s \partial_s u = - u ( 1 - u ) .
\end{equation}
This equation can be derived directly from (\ref{eqn:Smoluchowski}), or from
(\ref{eqn:SmoluchowskiContinuous}) if the initial number measure has a density
as in {\cite{MP2}}. Fix an initial point $s_0 \in \bar{\mathbb{C}}_+$. The
solution to (\ref{eqn:K=2PDE}) is given in terms of characteristics
$s ( \tau ; 0, s_0 )$ starting from $s_0$ at $\tau = 0$.
The characteristics satisfy
\begin{equation}
  \label{eqn:K=2Char0} \frac{ds}{d\tau} = s,
  \qquad s(0;0,s_0) = s_0 \in \bar{\mathbb{C}}_+,
\end{equation}
whose solution is
\begin{equation}
  \label{eqn:K=2Char} s ( \tau ; 0, s_0 ) = e^{\tau} s_0 \in
  \bar{\mathbb{C}}_+ .
\end{equation}
Note that the characteristics {\emph{are independent of the initial data
$u_0$}}. Geometrically, they are rays emanating from the origin at speed $e^{\tau}$.
In particular, the map $s_0 \mapsto s$ leaves the imaginary axis invariant.
Along characteristics we have
\[ \frac{du}{d \tau} = - u ( 1 - u ), \]
which can be integrated to yield the solution in terms of $s_0$:
\begin{equation}
  \label{eqn:K=2Soln} u ( \tau, s ) = e^{- \tau} \frac{u_0 ( s_0 )}{1 - u_0 (
  s_0 ) ( 1 - e^{- \tau} )} .
\end{equation}

The Fourier-Laplace transform of the self-similar solution,
\begin{equation}
  \label{eqn:K=2SolnSSExplicit} u_* ( s ) = \frac{1}{1 + s} .
\end{equation}
can also be expressed in terms of characteristics.
Fixing $s_0^* \in \bar{\mathbb{C}}_+$, the characteristics
$s^*(\tau; 0, s_0^*)$ corresponding to the solution of (\ref{eqn:K=2PDE})
with initial data $u_*$ are
\begin{equation}
  \label{eqn:K=2CharSS} s^* ( \tau ; 0, s_0^* ) = e^{\tau} s_0^* \in \bar{\mathbb{C}}_+.
\end{equation}
These characteristics are identical to (\ref{eqn:K=2Char}), as expected.
Since $u_* ( s )$ is a time-independent solution to (\ref{eqn:K=2PDE}),
\begin{equation}
  \label{eqn:K=2SolnSS} u_* ( s ) = e^{- \tau} \frac{u_* (
  s_0^* )}{1 - u_* ( s_0^* ) ( 1 - e^{- \tau} )} .
\end{equation}

{\emph{2. Estimates near the origin:}}
Now we utilize the additional finite moment assumption
\[ \mu_2 := \int_{( 0, \infty )} \hat{x}^2 \hat{n}_0 ( d x )
          = \int_{( 0, \infty )} x^2 n_0 ( d x ) < \infty \]
and that
\[ 2 = \int_{( 0, \infty )} \hat{x}^2 \hat{n}_* ( x ) dx
     = \int_{( 0, \infty )} x^2 n_* ( x ) dx \]
to obtain approximations for $u_0$ and $u_*$, respectively, near the
origin. By (\ref{eqn:TaylorApproximationZero}),
\begin{equation}
	\label{eqn:K=2EstimateOrigin}
	|u_0(s) - (1 - s)| \leq \frac{\mu_2}{2}|s|^2, \qquad
	|u_*(s) - (1 - s)| \leq |s|^2,
	\hspace{1em} s \in \bar{\mathbb{C}}_+.
\end{equation}

We can also estimate $u(\tau,s)$ near the origin using the time-dependent moment
\[ \hat{m}_2(\tau) = \int_{( 0, \infty )} \hat{x}^2 \hat{n}( \tau, d x )
     = \int_{( 0, \infty )} x^2 n( e^{\tau}, d x ) . \]
From (\ref{eqn:Smoluchowski}) (or (\ref{eqn:SmoluchowskiContinuous}) if $n(t,dx)$
has a density) it follows that $\partial_t{\hat{m}}_2(\tau) = 2 - \hat{m}_2(\tau)$, so
\[ \hat{m}_2(\tau) = 2 + (\mu_2 - 2)e^{-\tau}. \]
Therefore, (\ref{eqn:TaylorApproximationZero}) implies that for all
$\tau \geq 0$, 
\begin{equation}
  \label{eqn:K=2EstimateOrigin2}
	| u(\tau,s) - (1 - s) | \leq \frac{\hat{m}_2(\tau)}{2} |s|^2
	\leq \frac{\mu_2}{2} |s|^2, \hspace{1em} s \in \bar{\mathbb{C}}_+.
\end{equation}

{\emph{3. Propagation of estimates:}}
Consider the r.h.s. of (\ref{eqn:K=2Smoothing1}) with $T = \delta e^{\tau}$,
where $\delta > 0$ is a length scale characterizing the
region about the origin where the difference between $u_0$ and $u_*$ is small
by (\ref{eqn:K=2EstimateOrigin}). Explicitly, define
\begin{equation}
  2 \delta = \sqrt{1 + 2 \left( 1 + \frac{\mu_2}{2} \right)^{- 1}} - 1
\end{equation}
so that $\left( 1 + {\mu_2}/{2} \right) \delta = ( 1 + \delta )^{- 1} / 2$.
As expected, if $\mu_2 \rightarrow \infty$ then $\delta \rightarrow 0$.

{\emph{4. Backward characteristics:}}
Denote the common set of characteristics for $u$ and $u_*$ by $\sigma(\tau;0,\sigma_0)$
with $\sigma_0 \in \bar{\mathbb{C}}_+$. Now define the backward characteristic
$\sigma_0(\tau,\sigma_0)$ as the inverse of the forward mapping $\sigma_0
\mapsto \sigma(\tau; 0, \sigma_0)$. That is, for $\omega \in \bar{\mathbb{C}}_+$,
\begin{equation}
  \label{eqn:K=2BackwardCharacteristics}
  \sigma_0(\tau,\omega) = \sigma(0;\tau,\omega)
\end{equation}
Also note that by (\ref{eqn:K=2Char}) or (\ref{eqn:K=2CharSS}), the Jacobian of the
transformation $\sigma_0 \mapsto \sigma$ is
\begin{equation}
	\label{eqn:K=2CharJacobian}
  \frac{d\sigma}{d\sigma_0} = e^{\tau}.
\end{equation}
We now estimate (\ref{eqn:K=2Smoothing1}) by considering separately the
integral near and away from the origin.

{\emph{5. Bounds for $0 \leq | \sigma | \leq \delta e^{-\tau}$:}} Applying
(\ref{eqn:K=2EstimateOrigin}) and (\ref{eqn:K=2EstimateOrigin2}),
it is simple to see that

\begin{align}
	\left| \int_{0 \leq | \sigma | \leq \delta e^{-\tau} } 
	\frac{e^{\sigma \hat{x}}}{\sigma} ( u ( \tau, \sigma ) - u_* ( \sigma ) )
	d \sigma \right |
	& \leq 2 \int_0^{\delta e^{-\tau}}
	\frac{1}{|\sigma|} | u ( \tau, \sigma ) - u_* ( \sigma ) |
  d |\sigma|
	\nonumber\\
	& \leq \left(1 + \frac{\mu_2}{2} \right) \delta^2 e^{-2 \tau} .
	\label{eqn:K=2Bound1}
\end{align}

{\emph{6. Bounds for $\delta e^{-\tau} \leq | \sigma | \leq \delta e^{\tau}$:}}
Substituting (\ref{eqn:K=2Soln}), (\ref{eqn:K=2SolnSS}) into the r.h.s. of
(\ref{eqn:K=2Smoothing1}) and changing variables to $\sigma_0$ with (\ref{eqn:K=2Char})
and (\ref{eqn:K=2CharJacobian}) we have that
\begin{align}
	& \int_{\delta e^{-\tau} \leq | \sigma | \leq \delta e^{\tau} } 
	\frac{e^{\sigma \hat{x}}}{\sigma} ( u ( \tau, \sigma ) - u_* ( \sigma ) )
	d \sigma
	\nonumber \\
	& = \int_{\delta e^{-2 \tau} \leq | \sigma_0 | \leq \delta} \frac{e^{\sigma (
  \tau ; 0, \sigma_0 ) \hat{x}}}{\sigma ( \tau ; 0, \sigma_0 )}
  \frac{( u_0 ( \sigma_0 ) - u_* ( \sigma_0 ) )}{( 1 - u_* (
  \sigma_0 ) ( 1 - e^{- \tau} ) ) ( 1 - u_0 ( \sigma_0 ) ( 1 - e^{- \tau} ) )}
  d \sigma_0
\end{align}
Since $\sigma_0$ is pure imaginary in the region of integration,
(\ref{eqn:K=2SolnSSExplicit}) implies
\[ |1 - u_* ( \sigma_0 ) ( 1 - e^{- \tau} ) | = \left| \frac{e^{- \tau} +
   \sigma_0}{1 + \sigma_0} \right| \geq \frac{| \sigma_0 |}{1 + \delta} \]
and by (\ref{eqn:K=2EstimateOrigin}),
\[ |1 - u_0 ( \sigma_0 ) ( 1 - e^{- \tau} ) | \geq |1 - u_* ( \sigma_0 )
   ( 1 - e^{- \tau} ) | - \left( 1 + \frac{\mu_2}{2} \right) | \sigma_0 |^2
   \geq \frac{1}{2} \cdot \frac{| \sigma_0 |}{1 + \delta} . \]
Therefore,
\begin{align}
  \left| \int_{\delta e^{-2 \tau} \leq | \sigma_0 | \leq \delta} \frac{e^{\sigma (
  \tau ; 0, \sigma_0 ) \hat{x}}}{\sigma ( \tau ; 0, \sigma_0 )} \right. & \left.
  \frac{( u_0 ( \sigma_0 ) - u_* ( \sigma_0 ) )}{( 1 - u_* (
  \sigma_0 ) ( 1 - e^{- \tau} ) ) ( 1 - u_0 ( \sigma_0 ) ( 1 - e^{- \tau} ) )}
  d \sigma_0 \right| \nonumber\\
  & \leq 4 \left( 1 + \frac{\mu_2}{2} \right) ( 1 + \delta )^2 e^{- \tau}
  \int_{\delta e^{-2 \tau}}^{\delta} \frac{1}{| \sigma_0 |} d | \sigma_0 |
  \nonumber\\
  & = 8 \left( 1 + \frac{\mu_2}{2} \right) ( 1 + \delta )^2 \tau e^{-
  \tau}.
	\label{eqn:K=2Bound2}
\end{align}

{\emph{7. Rate of convergence:}} Substituting (\ref{eqn:K=2Bound1})
and (\ref{eqn:K=2Bound2}) into (\ref{eqn:K=2Smoothing1}), we obtain the desired result
(\ref{eqn:K=2Result}).

\subsection{Near-optimality of the rate for $K = 2$}
\label{subsec:K=2,Optimality}

We demonstrate that the exponential rate is achieved for monodisperse
initial data. Using the explicit solution given in {\cite{Aldous}} the
solution with $K = 2$ and initial data $n_0 ( d x )$ is
\begin{equation}
  \label{eqn:K=2OptimalitySoln} n ( t, d x ) = \sum_{j = 1}^{\infty} t^{- 2}
  \left( 1 - \frac{1}{t} \right)^{j - 1} \delta_j ( d x ), \qquad t \in
  [ 1, \infty ) .
\end{equation}
Here, $\delta_j ( d x )$ is the Dirac measure centered at $j \in \mathbb{N}$.
In similarity variables (\ref{eqn:RescaledSoln}) the distribution function for
$\hat{n} ( \tau, \hat{x} )$ corresponding to (\ref{eqn:K=2OptimalitySoln}) is
\[ F ( \tau, \hat{x} ) = e^{- \tau} \sum_{1 \leq j < e^{\tau} \hat{x}} ( 1 -
   e^{- \tau} )^{j - 1} = 1 - ( 1 - e^{- \tau} )^{\lceil e^{\tau} \hat{x}
   \rceil - 1} \]
where $\lceil \hat{x} \rceil = \min \{ j \in \mathbb{N} : j \geq \hat{x} \}$.
The limiting distribution is
\[ F_* ( \hat{x} ) = 1 - e^{- \hat{x}} . \]
We have that
\begin{equation}
  \label{eqn:K=2OptimalityDifference} F ( \tau, \hat{x} ) - F_* ( \hat{x}
  ) = e^{- \hat{x}} \left( 1 - e^{\alpha ( \tau, \hat{x} )} \right) = e^{-
  \hat{x}} \left( \alpha ( \tau, \hat{x} ) + O ( \alpha^2 ( \tau, \hat{x} ) )
  \right)
\end{equation}
where
\begin{equation}
  \label{eqn:K=2OptimalityExponent} \alpha ( \tau, \hat{x} ) = \hat{x} + ( \lceil
  e^{\tau} \hat{x} \rceil - 1 ) \log ( 1 - e^{- \tau} )
  = O \left( \left(1 + \frac{\hat{x}}{2} \right) e^{-\tau} \right)
\end{equation}
uniformly in $\hat{x}$ as $\tau \to \infty$. Since $0 \leq(1 + \hat{x}/2)e^{-\hat{x}}\leq 1$,
substituting (\ref{eqn:K=2OptimalityExponent}) in (\ref{eqn:K=2OptimalityDifference})
gives the desired result:
\[ \sup_{\hat{x} > 0} |F ( \tau, \hat{x} ) - F_* ( \hat{x} ) | = O ( e^{-\tau} ) . \]

\section{Rate of convergence for the additive kernel $K = x + y$}
\label{sec:K=x+y}

We now consider the case of the additive kernel:

\begin{theorem}
  \label{thm:K=x+y}Let $n_0$ be a positive measure, $\int_{( 0, \infty )} x n_0
  ( d x ) = \int_{( 0, \infty )} x^2 n_0 ( d x ) = 1$, and $\mu_3 =
  \int_{( 0, \infty )} x^3 n_0 ( d x ) < \infty$. With $\hat{n} ( t, d x )$ the
  rescaled solution to Smoluchowski's equation with $K = x + y$ and initial data
  $n_0 ( d x )$ let
  \[ F ( t, \hat{x} ) = \int_{( 0, \hat{x} ]} \hat{y} \hat{n} ( t, d \hat{y} ),
     \qquad F_* ( \hat{x} ) = \int_{( 0, \hat{x} ]} \hat{y}
     \hat{n}_{*, 1} ( \hat{y} ) d \hat{y} . \]
  Then for $t \in [ 0, \infty )$,
  \begin{equation}
    \label{eqn:K=x+yResult} \sup_{\hat{x} > 0} |F ( t, \hat{x} ) - F_* (
    \hat{x} ) | \leq C ( \mu_3 ) ( 1 + t ) e^{- t}
  \end{equation}
  where $C ( \mu_3 )$ is a constant that depends only on $\mu_3$.
\end{theorem}

\subsection{Proof for the additive kernel}
\label{subsec:K=x+y,Proof}

The analysis for the additive kernel is identical in spirit to that of the
constant kernel but is complicated by the fact that the characteristics of the
Fourier-Laplace transform are no longer rays. We use a contour deformation
argument to overcome this difficulty.

To start, for $s \in \mathbb{\bar{C}}_+$ define
\[ \varphi ( t, s ) = \int_{( 0, \infty )} ( 1 - e^{- s \hat{x}} ) \hat{n} (
   t, d \hat{x} ), \qquad \varphi_* ( s ) = \int_{( 0, \infty )} (
   1 - e^{- s \hat{x}} ) \hat{n}_* ( \hat{x} ) d \hat{x} \]
and let $\varphi_0 ( s ) = \varphi ( 0, s )$. We consider these quantities
instead of standard Laplace transforms since the number measure need not be
integrable near the origin. Another motivation is probabilistic: $\varphi$ and
$\varphi_*$ are \emph{characteristic exponents} of a subordinator with
no drift, as given by the L\'evy-Khintchine formula \cite{Bertoin}. We have that
the Fourier-Laplace transforms of the rescaled mass measures satisfy
\[ u ( t, s ) = \partial_s \varphi ( t, s ) = \int_{( 0, \infty )} e^{- s
   \hat{x}} \hat{x} \hat{n} ( t, d \hat{x} ), \qquad u_* ( s ) =
   \int_{( 0, \infty )} e^{- s \hat{x}} \hat{x} \hat{n}_* ( \hat{x} ) d \hat{x}
\]
with $u_0 ( s ) = u ( 0, s )$. By (\ref{eqn:Smoothing2}),
\begin{equation}
  \label{eqn:K=x+ySmoothing1} \sup_{\hat{x} > 0} |F ( t, \hat{x} ) - F_*
  ( \hat{x} ) | \leq \frac{1}{\pi} \sup_{\hat{x} > 0} \left| \int_{- i T}^{i
  T} \frac{e^{\sigma \hat{x}}}{\sigma} ( u ( t, \sigma ) - u_* ( \sigma )
  ) d \sigma \right| + \frac{24}{\pi T} .
\end{equation}

{\emph{1. Evolution of characteristics:}} The evolution of $\varphi$ and $u$
is given by
\begin{equation}
  \label{eqn:K=x+yPDE1} \partial_t \varphi + ( 2 s - \varphi ) \partial_s
  \varphi = \varphi
\end{equation}
\begin{equation}
  \label{eqn:K=x+yPDE2} \partial_t u + ( 2 s - \varphi ) \partial_s u = - u (
  1 - u ) .
\end{equation}
As with the constant kernel, these equations can be derived from
(\ref{eqn:Smoluchowski}), or from (\ref{eqn:SmoluchowskiContinuous}) if the
initial number measure has a density {\cite{MP2}}. Fixing an initial point $s_0
\in \bar{\mathbb{C}}_+$, the unique solution to
(\ref{eqn:K=x+yPDE1}-\ref{eqn:K=x+yPDE2}) is given in terms of the characteristics
$s ( t ; 0, s_0 )$ which satisfy
\begin{equation}
  \frac{d s}{d t} = 2 s - \varphi,
	\qquad s(0;0,s_0) = s_0 \in \bar{\mathbb{C}}_+.
\end{equation}
Along these characteristic curves we have that
\[ \frac{d \varphi}{d t} = \varphi, \qquad \frac{d u}{d t} = - u ( 1 - u
   ) \]
so that by integrating:
\begin{equation}
  \label{eqn:K=x+ySoln} \varphi ( t, s ) = e^t \varphi_0 ( s_0 ), \qquad
  u ( t, s ) = e^{- t} \frac{u_0 ( s_0 )}{1 - u_0 ( s_0 ) ( 1 - e^{- t} )} .
\end{equation}
These solutions are analytic, as shown in {\cite{MP1}}.
The explicit form for the characteristics $s_0 \mapsto s$, {\emph{which depend on the
initial data through $\varphi_0$}}, is then
\begin{equation}
  \label{eqn:K=x+yChar} s ( t ; 0, s_0 ) = e^{2 t} ( s_0 - \varphi_0 ( s_0 )
  ( 1 - e^{- t} ) ).
\end{equation}
This can be combined with the expression (\ref{eqn:K=x+ySoln}) for $\varphi$
to give
\begin{equation}
	\label{eqn:K=x+ySoln2}
  \frac{\varphi ( t, s )}{s} = e^{- t} \frac{( \varphi_0 ( s_0 ) / s_0 )}{1 -
  ( \varphi_0 ( s_0 ) / s_0 ) ( 1 - e^{- t} )} .
\end{equation}
Note that since $|u_0(s_0)| \leq 1$ and $|\varphi_0(s_0)| \leq |s_0|$ from
(\ref{eqn:K=x+ySoln}) and (\ref{eqn:K=x+ySoln2}) we have
\begin{equation}
  \label{eqn:K=x+yPhiBound} | \varphi ( t, s ) | \leq |s|, \qquad |u (
  t, s ) | \leq 1
\end{equation}
for $s(t;0,s_0)$ such that $s_0 \in \bar{\mathbb{C}}_+$ and $t\geq 0$.

The Fourier-Laplace transform
\begin{equation}
  \label{eqn:K=x+ySolnSSExplicit} u_* ( s ) = \frac{1}{\sqrt{1 + 2 s}}
\end{equation}
of the self-similar solution can also be stated in terms of characteristics.
In contrast to the constant kernel case,
these characteristics differ from those for the solution $u(t,s)$
with initial data $u_0$. Fixing $s^*_0 \in \bar{\mathbb{C}}$, this
distinct set of characteristics $s^*_0 \mapsto s^*$ satisfies
\begin{equation}
  \label{eqn:K=x+yCharSS} s^* ( t ; 0, s_0^* ) = e^{2 t} (
  s^*_0 - \varphi_* ( s_0^* ) ( 1 - e^{- t} ) )
\end{equation}
with $\varphi_* ( s ) = \sqrt{1 + 2 s} - 1$.
Since $u_* ( s )$ is a stationary solution along these characteristics,
we necessarily have that
\begin{equation}
  \label{eqn:K=x+ySolnSS} u_* ( s ) = e^{- t} \frac{u_* ( s_0^*
  )}{1 - u_* ( s_0^* ) ( 1 - e^{- t} )} .
\end{equation}

{\emph{2. Estimates near the origin:}}
We use the additional finite moment assumption
\[ \mu_3 = \int_{( 0, \infty )} \hat{x}^3 \hat{n}_0 ( d x )
         = \int_{( 0, \infty )} x^3 n_0 ( d x ) < \infty \]
along with
\[ 3 = \int_{( 0, \infty )} \hat{x}^3 \hat{n}_* ( x ) dx
     = \int_{( 0, \infty )} x^3 n_* ( x ) dx \]
to approximate $u_0$ and $u_*$ near the origin. By (\ref{eqn:TaylorApproximationZero}),
\begin{equation}
	\label{eqn:K=x+yEstimateOrigin1}
	|u_0(s) - (1 - s)| \leq \frac{\mu_3}{2}|s|^2, \qquad
	|u_*(s) - (1 - s)| \leq \frac{3}{2}|s|^2,
	\hspace{1em} s \in \bar{\mathbb{C}}_+.
\end{equation}
Since $u_0 = \partial_s \varphi_0$ implies
\[ \varphi_0(s) = s \int_0^1 u_0(\beta s) d\beta, \]
together with an analogous expression for $\varphi_*$ these estimates give that
\begin{equation}
	\label{eqn:K=x+yEstimateOrigin2}
	\left| \varphi_0(s) - \left(s - \frac{1}{2}s^2 \right) \right| \leq \frac{\mu_3}{6}|s|^3,
	\qquad \left| \varphi_*(s) - \left(s - \frac{1}{2}s^2 \right) \right| \leq \frac{1}{2}|s|^3,
	\hspace{1em} s \in \bar{\mathbb{C}}_+.
\end{equation}
This yields a bound on the distance between the characteristics $s, s^*$ starting at $\sigma_0 \in \bar{\mathbb{C}}_+$:
\begin{align}
  |s ( t ; 0, \sigma_0 ) - s^* ( t ; 0, \sigma_0 ) | & = e^{2 t} |
  \varphi_0 ( \sigma_0 ) - \varphi_* ( \sigma_0 ) | ( 1 - e^{- t} )
  \nonumber\\
  & \leq \left( \frac{1}{2} + \frac{\mu_3}{6} \right) ( 1 - e^{- t} ) e^{2
  t} | \sigma_0 |^3 .  \label{eqn:K=x+yCharDiff}
\end{align}
With (\ref{eqn:K=x+yPhiBound}) we also have that
\begin{equation}
  \label{eqn:K=x+yCharSSBound} |s^* ( t ; 0, \sigma_0 ) | = e^{2 t}
  \left| \sigma_0 - \varphi_* ( \sigma_0 ) ( 1 - e^{- t} ) \right| \leq
  2 e^{2 t} | \sigma_0 | .
\end{equation}

Finally, we estimate $u(t,s)$ itself near the origin using the
time-dependent moment
\[ \hat{m}_3(t) = \int_{( 0, \infty )} \hat{x}^3 \hat{n}( t, d x )
      = \int_{( 0, \infty )} x^3 n( t, d x ) . \]
It is straightforward from (\ref{eqn:Smoluchowski}) (or
(\ref{eqn:SmoluchowskiContinuous}) if $n(t,dx)$ has a density) that
$\partial_t {\hat{m}}_3(t) = 3 - \hat{m}_3(t)$, so
\[ \hat{m}_3(t) = 3 + (\mu_3 - 3)e^{-t}. \]
Then (\ref{eqn:TaylorApproximationZero}) implies that for all $t \geq 0$,
\begin{equation}
	\label{eqn:K=x+yEstimateOrigin3}
	| u(t,s) - (1 - s) | \leq \frac{\hat{m}_3(t)}{2} |s|^2
	   \leq \frac{\mu_3}{2} |s|^2.
\end{equation}

{\emph{3. Propagation of estimates:}}
Upon inspection of (\ref{eqn:K=x+yChar}) and (\ref{eqn:K=x+yCharSS}) we see
that the natural growth rate of both sets of characteristics is $e^t$ near the
origin, and $e^{2 t}$ at an $O ( 1 )$ distance away from the origin. Explicitly,
for small $\delta > 0$ note that
\[ s(t;0,i \delta) = e^{2t}(i \delta - \varphi_0(i \delta)(1-e^{-t}) )
   \approx e^{2t}\left( i \delta e^{-t} + \frac{1}{2}(i \delta)^2 (1 - e^{-t})
   \right). \]
Bending this point onto the imaginary axis, define
\begin{equation}
  \label{eqn:K=x+yDeformedContourEndpoint} T(t) = e^{2t}\left( \delta e^{-t}
  + \frac{1}{2} \delta^2 (1 - e^{-t}) \right).
\end{equation}
We will bound the r.h.s. of (\ref{eqn:K=x+ySmoothing1}) with this value for $T$,
where $\delta > 0$ is a length scale characterizing the
region about the origin where the difference between $u_0$ and $u_*$ is small by
(\ref{eqn:K=x+yEstimateOrigin1}). To begin, let
$\rho_3 = \max \{ \mu_3, 3 \}$ and define
\begin{equation}
  \label{eqn:K=x+yEstimateLengthScale} \delta = \sqrt{1 +
  \frac{1}{2} \left( \frac{1}{4} + \frac{\rho_3}{6} \right)^{- 1}} - 1
\end{equation}
so that $\left( \frac{1}{4} + \frac{\rho_3}{6} \right) \delta = \frac{1}{2} (
2 + \delta )^{- 1}$. Note that if $\mu_3 \rightarrow \infty$ then $\delta
\rightarrow 0$.

{\emph{4. Backward characteristics:}}
The main difficulty in the proof arises from the fact that characteristics
are no longer rays as in the constant kernel case. We now discuss properties of
the characteristic map which we will use later in the proof.

To begin, let $\Omega_t$ and $\Omega^*_t$ be the images of $\mathbb{C}_+$ under the
maps $s_0 \mapsto s(t;0,s_0)$ and $s_0^* \mapsto s^*(t;0,s_0^*)$ with $t\geq 0$.
By Lemma 3.3 in {\cite{MP2}}, the image $\Gamma_t = \partial \Omega_t$ of the
imaginary axis under $s_0 \mapsto s(t;0,s_0)$ is a curve which passes through the
origin and lies in the closed left half plane $\bar{\mathbb{C}}_-$. This is easily
seen by taking the real part of (\ref{eqn:K=x+yChar}) with $s = ik$ and noting that
\[ \text{Re } \varphi_0 (ik) = \int_{(0,\infty)} (1 - \cos(k x)) n_0(dx) \geq 0 . \]
Similarly, the image $\Gamma^*_t = \partial \Omega^*_t$ of the imaginary axis under
$s_0^* \mapsto s^*(t;0,s_0^*)$ passes through the origin and otherwise lies in the
\emph{open} left half plane $\mathbb{C}_-$ due to the continuity of $n_*(x)$.
Furthermore, it has been shown that $\Gamma_t$ and $\Gamma^*_t$ are
curves on which $\text{Re } s$ is a function of $\text{Im } s$.

We also have by {\cite{MP2}}, Lemma 3.3(iii) that for $t \geq 0$,
$s_0 \mapsto s(t;0,s_0)$ is one-to-one from $\bar{\mathbb{C}}_+$ onto $\bar{\Omega}_t$.
It is analytic in $\mathbb{C}_+$, and its inverse map $s \mapsto s_0$ is analytic in
$\Omega_t$. In particular, this yields the analyticity of $u(t,s)$ in $\Omega_t$
(and analyticity of $u_*(s)$ in $\Omega^*_t$ by substituting $s$ with $s^*$). As
the positive real axis is invariant under $s_0 \mapsto s(t;0,s_0)$, the continuity
of the mapping implies that characteristics starting in the upper or lower half of
$\mathbb{C}_+$ remain in the upper or lower half of the complex plane, respectively.

Next, define the backward characteristics $s_0(t,s)$ and $s_0^*(t,s)$ as inverses
of the forward mappings $s_0 \mapsto s(t; 0, s_0)$ and
$s_0^* \mapsto s^*(t; 0, s_0^*)$, respectively. Explicitly, for
$\omega \in \bar{\Omega}_t$ and $\omega^* \in \bar{\Omega}^*_t$,
\begin{equation}
  \label{eqn:K=x+yBackwardCharacteristics}
  s_0(t,\omega) = s(0;t,\omega), \qquad s_0^*(t,\omega^*) = s^*(0;t,\omega^*).
\end{equation}
By (\ref{eqn:K=x+yChar}) and (\ref{eqn:K=x+yCharSS}), the Jacobians of the
transformations $s_0 \mapsto s$ and $s \mapsto s^*$ are
\begin{equation}
	\label{eqn:K=x+yCharJacobian}
	\frac{d s}{d s_0} = e^{2 t} ( 1 - u_0 ( s_0 ) ( 1 - e^{- t} ) ), \qquad
	 \frac{d s^*}{d s_0^*} = e^{2 t} ( 1 -  u_* ( s_0^* ) ( 1 - e^{- t} ) ) .
\end{equation}

We will now estimate (\ref{eqn:K=x+ySmoothing1}) by considering separately the integral
near and away from the origin. In doing so, we use the following convention for
the remainder of the proof. The independent variable $\sigma$ is used
solely for points on the imaginary axis in the $s,s^*$-plane, while $\sigma_0$
is used analogously in the $s_0,s_0^*$-plane.

{\emph{5. Bounds for $0 \leq | \sigma | \leq \delta e^{-t}$:}} Applying
(\ref{eqn:K=x+yEstimateOrigin1}) and (\ref{eqn:K=x+yEstimateOrigin3}), we get that

\begin{align}
	\left| \int_{0 \leq | \sigma | \leq \delta e^{-t} } 
	\frac{e^{\sigma \hat{x}}}{\sigma} ( u ( t, \sigma ) - u_* ( \sigma ) )
	d \sigma \right |
	& \leq 2 \int_0^{\delta e^{-t}}
	\frac{1}{|\sigma|} | u ( t, \sigma ) - u_* ( \sigma ) |
  d |\sigma|
	\nonumber\\
	& \leq \left(\frac12 + \frac{\mu_3}{6} \right) \delta^2 e^{-2t} .
	\label{eqn:K=x+yBound1}
\end{align}

{\emph{6. Bounds for $\delta e^{-t} \leq | \sigma | \leq T(t)$:}}
Since the flow of characteristics no longer leaves the imaginary axis
invariant, we use a contour deformation argument (see {\cite{MP2}} for
a similar argument).

{\emph{6a. Contour deformation:}}
Without loss of generality, we work in the upper half of the complex plane.
Bounds for quantities in the lower half plane are identical since
$u(t,\bar{s}) = \overline{u(t,s)}$ and $u_*(\bar{s}) = \overline{u_*(s)}$.
Take
\[ \mathcal{A}' = i \delta e^{-2t}, \qquad \mathcal{D}' = i \delta \]
to be the endpoints of a piece of the imaginary axis in the $s_0,s_0^*$-
plane (see Figure \ref{fig:K=x+yContourDeformation}). Let $AD$ and $A^*D^*$
be the images of $\mathcal{A}'\mathcal{D}'$ under the mappings
$s_0 \mapsto s(t;0,s_0)$ and $s_0^* \mapsto s^*(t;0,s_0^*)$.
As discussed in step 4, $AD \subset \Gamma_t$ and $A^*D^* \subset \Gamma^*_t$
lie in the left half of the complex plane.

The part of (\ref{eqn:K=x+ySmoothing1}) that remains to be estimated is
an integral over the imaginary axis in the $s,s^*$-plane, bounded
away from the origin, with endpoints
\[ \mathcal{B} = i \delta e^{-t}, \qquad \mathcal{C} = i T(t). \]
To do so, split the integral into its individual components
\begin{equation}
  \int_{\mathcal{B}\mathcal{C}} \frac{e^{\sigma \hat{x}}}{\sigma} ( u ( t, \sigma ) -
  u_* ( \sigma ) ) d \sigma = \int_{\mathcal{B}\mathcal{C}} \frac{e^{\sigma
  \hat{x}}}{\sigma} u ( t, \sigma ) d \sigma - \int_{\mathcal{B}\mathcal{C}} \frac{e^{\sigma
  \hat{x}}}{\sigma} u_* ( \sigma ) d \sigma .
\end{equation}
Since $u$ and $u_*$ are analytic in $\Omega_t$ and $\Omega^*_t$, the integrands
are analytic away from the origin and we can apply Cauchy's theorem:
\begin{equation}
  \label{eqn:K=x+yContourDeformation1}
	\int_{\mathcal{B}\mathcal{C}} \frac{e^{\sigma \hat{x}}}{\sigma} u ( t, \sigma ) d \sigma
	= \int_{AD} \frac{e^{s \hat{x}}}{s} u ( t, s ) d s
	+ \int_{\mathcal{B}A} \frac{e^{s \hat{x}}}{s} u ( t, s ) d s
	+ \int_{D\mathcal{C}} \frac{e^{s \hat{x}}}{s} u ( t, s ) d s,
\end{equation}
and
\begin{equation}
  \label{eqn:K=x+yContourDeformation2}
	\int_{\mathcal{B}\mathcal{C}} \frac{e^{\sigma \hat{x}}}{\sigma} u_* ( \sigma ) d \sigma
	= \int_{A^*D^*} \frac{e^{s^*  \hat{x}}}{s^*} u_* ( s^* ) d s^*
	+ \int_{\mathcal{B}A^*} \frac{e^{s^*  \hat{x}}}{s^*} u_* ( s^* ) d s^*
	+ \int_{D^*\mathcal{C}} \frac{e^{s^*  \hat{x}}}{s^*} u_* ( s^* ) d s^* .
\end{equation}
The error terms are integrals over contours that connect the endpoints of $\mathcal{B}\mathcal{C}$
to those of $AD$ and $A^*D^*$, respectively. We take $\mathcal{B}A$ and $\mathcal{B}A^*$ to
be straight line contours. Additionally, let $D\mathcal{C}$ and $ D^*\mathcal{C}$ be
such that their backward images
\begin{equation}
  \mathcal{D}'C' = s_0 ( t, D\mathcal{C} ),
	\qquad \mathcal{D}'{C^*}' = s_0^* (t, D^*\mathcal{C})
\end{equation}
are straight lines. We will estimate the integrals over these contours in (\ref{eqn:K=x+yContourDeformation1}-\ref{eqn:K=x+yContourDeformation2}) at the end of the
section.

\begin{figure}[htb]
	\center{\includegraphics[scale=0.6]{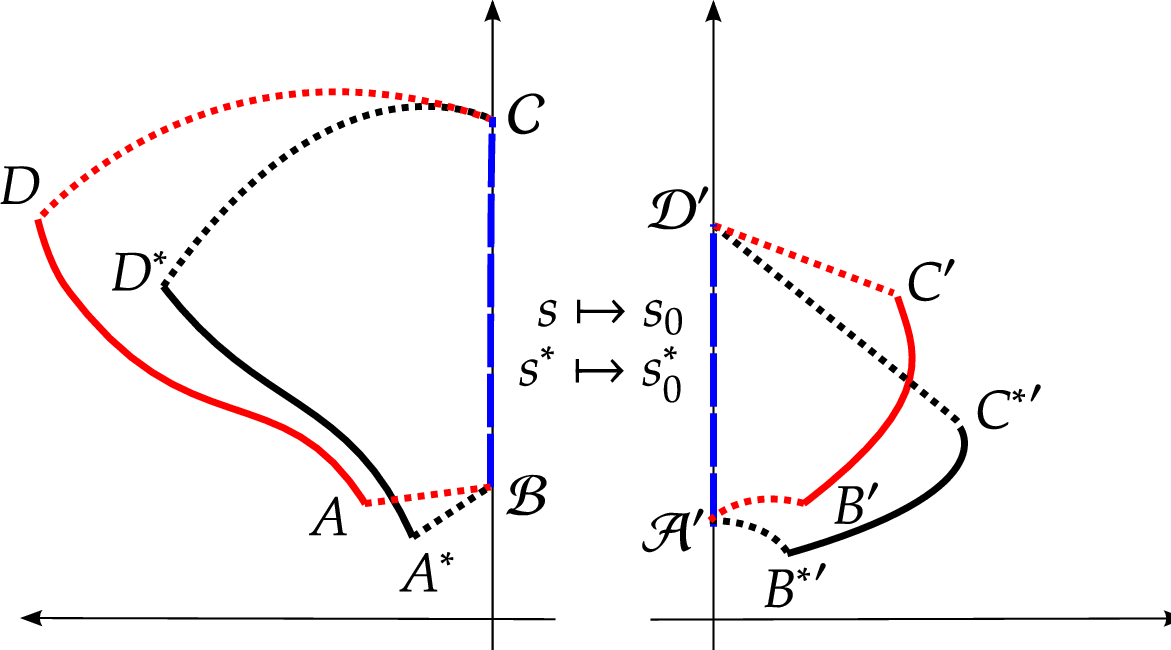}}
	\caption{\label{fig:K=x+yContourDeformation}
	Endpoints of contours on imaginary axis are
	$\mathcal{B} = i \delta e^{-t}, \mathcal{C} = i T(t)$
	and $\mathcal{A}' = i \delta e^{-2t}, \mathcal{D}' = i \delta$.
	On the left,
	$A = s(t;0, \mathcal{A}'), A^* = s^*(t;0, \mathcal{A}'),
	D = s(t;0, \mathcal{D}'), D^* = s^*(t;0, \mathcal{D}')$.
  On the right,
	$B' = s_0(t, \mathcal{B}), {B^*}' = s_0^*(t, \mathcal{B}),
	C' = s_0(t, \mathcal{C}), {C^*}' = s_0^*(t, \mathcal{C})$.
	}
\end{figure}

Making the change of variables $s \mapsto s_0$ and
$s^* \mapsto s_0^*$ and using (\ref{eqn:K=x+ySoln}),
(\ref{eqn:K=x+ySolnSS}), and (\ref{eqn:K=x+yCharJacobian}), we have:
\begin{align}
  & \int_{AD} \frac{e^{s \hat{x}}}{s} u ( t, s ) d s - \int_{A^*D^*}
  \frac{e^{s^*  \hat{x}}}{s^*} u_* ( s^* ) d s^*
  \nonumber\\
  & = e^{- t} \int_{\mathcal{A}'\mathcal{D}'} \left( e^{s ( t ; 0, \sigma_0 ) \hat{x}}
  \frac{e^{2 t}}{s ( t ; 0, \sigma_0 )} u_0 ( \sigma_0 ) - e^{s^* ( t ;
  0, \sigma_0 ) \hat{x}} \frac{e^{2 t}}{s^* ( t ; 0, \sigma_0 )} u_*
  ( \sigma_0 ) \right) d \sigma_0 .  \label{eqn:K=x+yMainEstimate}
\end{align}
Using the basic inequality $|a_1 b_1 - a_2 b_2 | \leq |b_1 - b_2 ||a_1 | +
|a_1 - a_2 ||b_2 |$, we can bound the magnitude of
(\ref{eqn:K=x+yMainEstimate}) by the sum $\Lambda_1 + \Lambda_2$, where
\begin{equation}
  \label{eqn:K=x+yMainTerm1} \Lambda_1 = e^{-t} \int_{\delta e^{-2t}}^{\delta}
  \left| \frac{e^{2 t}}{s} u_0 ( \sigma_0 ) - \frac{e^{2 t}}{s^*} u_* (
  \sigma_0 ) \right| |e^{s \hat{x}} |d| \sigma_0 |
\end{equation}
\begin{equation}
  \label{eqn:K=x+yMainTerm2} \Lambda_2 = e^{-t} \int_{\delta e^{-2t}}^{\delta}
  |e^{s \hat{x}} - e^{s^* \hat{x}} | \left| \frac{e^{2 t}}{s^*} u_* (
  \sigma_0 ) \right| d| \sigma_0 |
\end{equation}
We now evaluate each of these terms. For $\Lambda_1$ the argument is nearly
identical to that presented in Section \ref{sec:K=2}.

{\emph{6b. Bounds for $\Lambda_1$:}}
Since $s \in \bar{\mathbb{C}}_-$, $|e^{s \hat{x}} | \leq
1$. Using the explicit form (\ref{eqn:K=x+yChar}), (\ref{eqn:K=x+yCharSS}) of
the characteristics, we rearrange terms in the integrand of
(\ref{eqn:K=x+yMainTerm1}) to get
\begin{align}
  \left| \frac{e^{2 t}}{s} u_0 ( \sigma_0 ) \right. & \left. - \frac{e^{2
  t}}{s^*} u_* ( \sigma_0 ) \right| \nonumber\\
  & = \left| \frac{1}{\sigma_0} \frac{\left( \frac{e^{- 2 t}
  s^*}{\sigma_0} u_0 ( \sigma_0 ) - \frac{e^{- 2 t} s}{\sigma_0} u_*
  ( \sigma_0 ) \right)}{\left( 1 - \frac{\varphi_* ( \sigma_0
  )}{\sigma_0} ( 1 - e^{- t} ) \right) \left( 1 - \frac{\varphi_0 ( \sigma_0
  )}{\sigma_0} ( 1 - e^{- t} ) \right)} \right| . 
  \label{eqn:K=x+yMainTerm1Rewrite}
\end{align}
Notice that here $\varphi_0 ( \sigma_0 ) / \sigma_0$ and $\varphi_* ( \sigma_0
) / \sigma_0$ play the same role as $u_0 ( \sigma_0 )$ and $u_* (
\sigma_0 )$ do in (\ref{eqn:K=2Bound2}).
Using that $|u_* | \leq 1$ and the estimates (\ref{eqn:K=x+yCharDiff}),
(\ref{eqn:K=x+yCharSSBound}), the numerator in (\ref{eqn:K=x+yMainTerm1Rewrite})
satisfies
\begin{align}
  \left| \frac{e^{- 2 t} s^*}{\sigma_0} u_0 ( \sigma_0 ) \right. &
  \left. - \frac{e^{- 2 t} s}{\sigma_0} u_* ( \sigma_0 ) \right|
  \nonumber\\
  & \leq \left| \frac{e^{- 2 t} s^*}{\sigma_0} - \frac{e^{- 2 t}
  s}{\sigma_0} \right| |u_* ( \sigma_0 ) | + |u_0 ( \sigma_0 ) - u_*
  ( \sigma_0 ) | \left| \frac{e^{- 2 t} s^*}{\sigma_0} \right|
  \nonumber\\
  & \leq 7 \left( \frac{1}{2} + \frac{\mu_3}{6} \right) | \sigma_0 |^2 . 
  \label{eqn:K=x+yMainTerm1Num}
\end{align}

To estimate the denominator of (\ref{eqn:K=x+yMainTerm1Rewrite}), we replace
$\varphi_0(\sigma_0)$ and $\varphi_*(\sigma_0)$ by $2\sigma_0 /(2 + \sigma_0)$.
Since
\[ \left| \frac{2s}{2 + s}  - \left(s - \frac{1}{2}s^2 \right) \right| \leq
   \frac{1}{4}|s|^3, \hspace{1em} s \in \bar{\mathbb{C}}_+,  \]
with (\ref{eqn:K=x+yPhiBound}), (\ref{eqn:K=x+yEstimateOrigin2}), and
(\ref{eqn:K=x+yEstimateLengthScale}) this implies that
\begin{align}
  \left| 1 - \frac{\varphi_0 ( \sigma_0 )}{\sigma_0} ( 1 - e^{- t} ) \right|
  \geq \left| \frac{2 e^{- t} + \sigma_0}{2 + \sigma_0} \right| & - \left(
  \frac{1}{4} + \frac{\mu_3}{6} \right) | \sigma_0 |^2 \nonumber\\
  & \geq \frac{| \sigma_0 |}{2 + \delta} - \left( \frac{1}{4} +
  \frac{\mu_3}{6} \right) | \sigma_0 |^2 \geq \frac{1}{2} \cdot \frac{|
  \sigma_0 |}{2 + \delta} .  \label{eqn:K=x+yMainTerm1DeltaBound2}
\end{align}
The same argument holds with $\varphi_0$ replaced by $\varphi_*$ and
$\mu_3$ replaced by $3$, but $\delta$ unchanged:
\begin{align}
  \left| 1 - \frac{\varphi_* ( \sigma_0 )}{\sigma_0} ( 1 - e^{- t} )
  \right| \geq \frac{1}{2} \cdot \frac{| \sigma_0
  |}{2 + \delta} . \label{eqn:K=x+yMainTerm1DeltaBound}
\end{align}
Therefore, using
(\ref{eqn:K=x+yMainTerm1DeltaBound}-\ref{eqn:K=x+yMainTerm1DeltaBound2}) and
(\ref{eqn:K=x+yMainTerm1Num}) in (\ref{eqn:K=x+yMainTerm1Rewrite}) yields
\begin{align}
  e^{- t} \int_{\delta e^{- 2t} \leq | \sigma_0 | \leq \delta} \left|
  \frac{e^{2 t}}{s} u_0 ( \sigma_0 ) \right. & \left. - \frac{e^{2
  t}}{s^*} u_* ( \sigma_0 ) \right| |e^{s \hat{x}} |d|
  \sigma_0 | \nonumber\\
  & \leq 28 \left( \frac{1}{2} + \frac{\mu_3}{6} \right) ( 2 +
  \delta )^2 e^{- t} \int_{\delta e^{- 2t}}^{\delta} \frac{1}{| \sigma_0 |} d |
  \sigma_0 | \nonumber\\
  & = 56 \left( \frac{1}{2} + \frac{\mu_3}{6} \right) ( 2 + \delta
  )^2 t e^{- t} .  \label{eqn:K=x+yMainTerm1Estimate2}
\end{align}

{\emph{6c. Bounds for $\Lambda_2$:}} Now we consider $\Lambda_2$. Using
$|u_* | \leq 1$ and (\ref{eqn:K=x+yCharSS}) we can rewrite this as
\begin{equation}
  \label{eqn:K=x+yMainTerm2Rewrite} \Lambda_2 \leq e^{- t}
  \int_{\delta e^{-2t}}^{\delta}
  \frac{|e^{s \hat{x}} - e^{s^* \hat{x}} |}{| \sigma_0 |}
  \frac{1}{\left| 1 - \frac{\varphi_* ( \sigma_0 )}{\sigma_0} ( 1 - e^{-
  t} ) \right|} d| \sigma_0 |.
\end{equation}
As noted in step 4, since $\sigma_0 = ik$ with $k \in \mathbb{R}$ nonzero, we have that
$\text{Re } s^* < 0$. This implies that $\sup_{\hat{x} > 0} | \hat{x} e^{s^* \hat{x}} |$
is bounded above by $1 / |\text{Re } s^*|$. Then
\begin{align}
  |e^{s \hat{x}} - e^{s^* \hat{x}}|
	= |e^{s^* \hat{x}}||e^{(s-s^*) \hat{x}} - 1|
  & \leq |s-s^*|\sup_{\hat{x} > 0} | \hat{x} e^{s^* \hat{x}} | \nonumber\\
  & \leq \frac{|s-s^*|}{|\text{Re } s^*|} \leq \left( \frac{1}{2} + \frac{\mu_3}{6} \right)
	\frac{|\sigma_0|^3}{|\text{Re } \varphi_* ( \sigma_0 ) |} , \label{eqn:K=x+yMainTerm2Num}
\end{align}
where the last inequality follows from (\ref{eqn:K=x+yCharSS}) and (\ref{eqn:K=x+yCharDiff}).
We now make use of the explicit form $\varphi_* ( \sigma_0 ) = \sqrt{1 + 2
\sigma_0} - 1$ obtained from (\ref{eqn:K=x+ySolnSSExplicit}). Since
$re^{i\theta} = 1 + 2ik \in \mathbb{C}_+$ with $r = \sqrt{1 + 4k^2}$ and $|\theta| < \pi/2$,
\[ \text{Re } \sqrt{1 + 2ik} = \sqrt{r} \cos \left( \frac{\theta}{2} \right)
= \sqrt{\frac{1}{2} (r+1)} \]
by using a double-angle formula and $r \cos(\theta) = 1$. To summarize,
\[ \text{Re } \varphi_*(\sigma_0)
= \sqrt{\frac{1}{2} \left( 1 + \sqrt{1 + 4|\sigma_0|^2} \right)} - 1 . \]
Defining $g(y) = \sqrt{\frac12 \left(1 + \sqrt{1+y} \right)} - 1$,
a routine calculation shows that $g'''(y) > 0$ for all $y>0$. It follows by Taylor's theorem
that $g(y) \geq y/8 - 5y^2/128$. Therefore, if $|\sigma_0| \leq \delta$
then $\text{Re } \varphi_*(\sigma_0) = g \left( 4|\sigma_0|^2 \right)$ satisfies
\begin{equation}
  \label{eqn:K=x+yMainTerm2EstimateRealPart} | \text{Re } \varphi_*(\sigma_0 ) |
	\geq \left(\frac{1}{2} - \frac{5}{8}\delta^2 \right) | \sigma_0^2 |
	\geq \frac{1}{4} | \sigma_0^2 |
\end{equation}
since $\delta^2 < 2/5$ by (\ref{eqn:K=x+yEstimateLengthScale}).
Combining (\ref{eqn:K=x+yMainTerm2Num}) and (\ref{eqn:K=x+yMainTerm2EstimateRealPart})
yields
\begin{equation}
  \label{eqn:K=x+yMainTerm2PreEstimate} \frac{|e^{s \hat{x}} - e^{s^*
  \hat{x}} |}{| \sigma_0 |} \leq 4 \left( \frac{1}{2} + \frac{\mu_3}{6}
  \right) .
\end{equation}

For the remaining term in the integrand of (\ref{eqn:K=x+yMainTerm2Rewrite}) we use the
estimate (\ref{eqn:K=x+yMainTerm1DeltaBound}). Putting this together with
(\ref{eqn:K=x+yMainTerm2PreEstimate}) we conclude that
\begin{align}
  e^{- t} \int_{\delta e^{- 2t}}^{\delta} \frac{|e^{s \hat{x}} - e^{s^*
  \hat{x}} |}{| \sigma_0 |} & \frac{1}{\left| 1 -
  \frac{\varphi_* ( \sigma_0 )}{\sigma_0} ( 1 - e^{- t} ) \right|}
  d | \sigma_0 | \nonumber\\
  & \leq 8 \left( \frac{1}{2} + \frac{\mu_3}{6} \right) ( 2 +
  \delta ) e^{- t} \int_{\delta e^{- 2t}}^{\delta} \frac{1}{| \sigma_0 |} d |
  \sigma_0 | \nonumber\\
  & = 16 \left( \frac{1}{2} + \frac{\mu_3}{6} \right) ( 2 + \delta
  ) t e^{- t} .  \label{eqn:K=x+yMainTerm2Estimate2}
\end{align}

{\emph{6d. Bounds for $\Lambda_1 + \Lambda_2$:}}
Combining (\ref{eqn:K=x+yMainTerm1Estimate2}) and
(\ref{eqn:K=x+yMainTerm2Estimate2}), we obtain a bound for
(\ref{eqn:K=x+yMainTerm1}-\ref{eqn:K=x+yMainTerm2}):
\begin{equation}
  \label{eqn:K=x+yLambdaEstimates} \Lambda_1 + \Lambda_2 \leq 64 \left(
  \frac{1}{2} + \frac{\mu_3}{6} \right) ( 2 + \delta )^2 te^{- t} .
\end{equation}

{\emph{7. Error terms:}}
Lastly, we bound the two error terms in (\ref{eqn:K=x+yContourDeformation1})
obtained from the contour deformation argument. The corresponding error terms
in (\ref{eqn:K=x+yContourDeformation2}) are bounded by an argument identical
to the one given below with only constants differing.

First we consider the term
\begin{equation}
	\label{eqn:K=x+yErrorTerm1}
	\left| \int_{\mathcal{B}A} \frac{e^{s \hat{x}}}{s} u ( t, s ) d s \right|
	\leq |\mathcal{B}A| \sup_{\mathcal{B}A} \left( |s|^{- 1} \right) .
\end{equation}
Here we have used that for $s \in \bar{\Gamma}_t$, both $|e^{s\hat{x}}|$ and
$|u ( t, s )|$ are bounded by 1. From the explicit expression for the endpoints
we have that
\[	|\mathcal{B}A| = | e^{2t}(i \delta e^{-2t} - \varphi_0(i \delta e^{-2t})(1 - e^{-t}))
	- i \delta e^{-t} | .\]
Substituting $\varphi_0(i \delta e^{-2t})$ with $i \delta e^{-2t}$ and using the bound
\begin{equation}
	\label{eqn:K=x+yError1Bound1}
	|\varphi_0(s_0) - s_0| \leq \frac12 |s_0|^2, \hspace{1em} s_0 \in \bar{\mathbb{C}}_+
\end{equation}
derived from (\ref{eqn:TaylorApproximationZero}), we find that
$|\mathcal{B}A| \leq \frac12 \delta^2 e^{-2t}$.

For the remaining term, write $s \in \mathcal{B}A$ as an interpolation of the
endpoints:
\[ s = \beta (e^{2t}(i \delta e^{-2t} - \varphi_0(i \delta e^{-2t})(1 - e^{-t})))
			+ (1 - \beta)i \delta e^{-t}, \qquad \beta \in [0,1] . \]
Again, substituting $\varphi_0(i \delta e^{-2t})$ by $i \delta e^{-2t}$ and using
(\ref{eqn:K=x+yError1Bound1}) implies that for $s \in \mathcal{B}A$,
\[ |s| \geq \delta e^{-t} - \frac{\beta}{2}\delta^2 e^{-2t} \geq \frac12 \delta e^{-t} . \]
Putting this together with the previous bound, we conclude that
\begin{equation}
  \label{eqn:K=x+yErrorEstimate1} \left| \int_{\mathcal{B}A} \frac{e^{s
  \hat{x}}}{s} u ( t, s ) d s \right| \leq \delta e^{-t} .
\end{equation}

Now we estimate the second error term. Changing variables from $s$ to $s_0$,
\begin{align}
  \left| \int_{D\mathcal{C}} \frac{e^{s \hat{x}}}{s} u ( t, s ) d s \right|
  & = e^{- t} \left| \int_{\mathcal{D}'C'} \frac{e^{s ( t ; 0, s_0 ) \hat{x}}}{s_0}
  \frac{u_0 ( s_0 )}{1 - \frac{\varphi_0 ( s_0 )}{s_0} ( 1 - e^{- t} )} d s_0
  \right| \nonumber\\
  & \leq e^{- t} |\mathcal{D}'C'| \sup_{\mathcal{D}'C'} \left( |s_0 |^{- 1} \right) \cdot
  \sup_{\mathcal{D}'C'} \left( \left| 1 - \frac{\varphi_0 ( s_0 )}{s_0} ( 1 - e^{- t} )
  \right|^{- 1} \right) . \label{eqn:K=x+yErrorTerm2}
\end{align}
The endpoints of the line segment $\mathcal{D}'C'$
are $i \delta$ and $\pi_0 := s_0 (t, i T ( t ) )$. By (\ref{eqn:K=x+yChar}) and
(\ref{eqn:K=x+yDeformedContourEndpoint}), the latter satisfies
\begin{equation}
  \label{eqn:K=x+yEndpointEqn}
  i\left(\delta e^{-t} + \frac12 \delta^2 (1 - e^{-t}) \right)
			= \pi_0 - \varphi_0(\pi_0)( 1 - e^{-t} ) .
\end{equation}
Properties of the characteristic map (see step 4) imply that this equation
has a unique solution in the upper half of $\bar{\mathbb{C}}_+$.
Rewriting the r.h.s. as $\pi_0 e^{-t} + (\pi_0 - \varphi_0(\pi_0))(1 - e^{-t})$,
we take absolute values and use (\ref{eqn:K=x+yError1Bound1}) to get that
\[ \delta e^{-t} + \frac12 \delta^2 (1 - e^{-t})
			\leq |\pi_0| e^{-t} + \frac12 |\pi_0|^2 (1 - e^{-t}). \]
Since $h(y) = \beta y + (1 - \beta)\frac12 y^2$ is strictly increasing for positive
$y$ and any $\beta \in [0,1]$, this implies that $|\pi_0| \geq \delta$.
The endpoints of the line $\mathcal{D}'C'$ are therefore at least a distance
$\delta$ away from the origin, so
$\sup_{\mathcal{D}'C'} |s_0 |^{-1} \leq \sqrt{2}\delta^{- 1} = : C_1 ( \delta ) $.

We also have an upper bound for $|\pi_0 |$ as follows. Since the mass measure
$\hat{x}\hat{n}_0(d \hat{x})$ is not concentrated at $\hat{x} = 0$, $|u_0|$ is
bounded away from $1$ in an annulus about the origin. Then there exists an
$\varepsilon = \varepsilon ( \delta ) > 0$ such that for all $s_0 \in \mathcal{D}'C'$,
\[ \left| \frac{\varphi_0 ( s_0 )}{s_0} \right|
  \leq \int_0^1 |u_0 ( \beta s_0 ) |d \beta \leq 1 - \varepsilon . \]
This shows that
\begin{equation}
  \label{eqn:K=x+yError2Bound1}
	\sup_{\mathcal{D}'C'} \left( \left| 1 - \frac{\varphi_0 ( s_0 )}{s_0} ( 1 - e^{- t} )
   \right|^{- 1} \right) \leq \varepsilon^{- 1} .
\end{equation}
Rewriting the r.h.s. of (\ref{eqn:K=x+yEndpointEqn}) as
$\pi_0 \left(1 - (\varphi_0(\pi_0) / \pi_0)(1 - e^{-t}) \right)$,
we find
\[ |\pi_0| = \left(\delta e^{-t} + \frac12 \delta^2 (1 - e^{-t}) \right)
			\left|1 - \frac{\varphi_0(\pi_0)}{\pi_0}(1 - e^{-t}) \right|^{-1}
			\leq \frac{\delta + \frac12 \delta^2}{\varepsilon} . \]
Thus $|\mathcal{D}'C'|$ is bounded above by some constant $C_2 ( \delta )$.

Finally, putting these estimates together in
(\ref{eqn:K=x+yErrorTerm2}) yields
\begin{equation}
  \label{eqn:K=x+yErrorEstimate2} \left| \int_{D\mathcal{C}} \frac{e^{s
  \hat{x}}}{s} u ( t, s ) d s \right| \leq C_3 ( \delta ) e^{- t} .
\end{equation}
where $C_3 : = C_1 C_2 \varepsilon^{- 1}$ depends on the initial data
$u_0$ only through $\mu_3$.

{\emph{8. Rate of convergence:}} Using (\ref{eqn:K=x+yBound1}) and
(\ref{eqn:K=x+yLambdaEstimates}) in (\ref{eqn:K=x+ySmoothing1})
with the bounds (\ref{eqn:K=x+yErrorEstimate1}) and \ref{eqn:K=x+yErrorEstimate2}) for
the error terms, we obtain the desired result (\ref{eqn:K=x+yResult}).

\subsection{Monodisperse initial data with $K = x + y$}
\label{subsec:K=x+y,Optimality}

As shown in {\cite{Aldous}}, the explicit solution with $K = x + y$ and
initial data $n_0 ( d x ) = \delta_1 ( d x )$ can be given in terms of the
Borel-Tanner distribution
\[ B ( \lambda, j ) = \lambda^{j - 1} \frac{j^{j - 1} e^{- \lambda j}}{j!},
   \qquad 0 \leq \lambda \leq 1 \]
by
\begin{equation}
  \label{eqn:K=x+yOptimalitySoln} n ( t, d x ) = \sum_{j = 1}^{\infty} e^{- t}
  B ( 1 - e^{- t}, j ) \delta_j ( d x ), \qquad t \in [ 0, \infty ) .
\end{equation}
In the probabilistic setting, $B ( \lambda, j )$ is the distribution of the total
population size $X_{\lambda}$ of a Galton-Watson branching process with one
progenitor and Poisson($\lambda$) offspring distribution. Note that
\[ \mathbb{E} [ X_{\lambda} ] = \sum_{j = 1}^{\infty} j B ( \lambda, j ) = ( 1 -
   \lambda )^{- 1} . \]
In similarity variables (\ref{eqn:RescaledSoln}) the distribution function for
$\hat{x} \hat{n} ( t, d x )$ corresponding to (\ref{eqn:K=x+yOptimalitySoln})
is
\begin{equation}
  \label{eqn:K=x+yOptimalityDistFcn} F ( t, \hat{x} ) = e^{- 2 t} \sum_{1 \leq
  j < e^{2 t} \hat{x}} e^t ( 1 - e^{- t} )^{j - 1} \frac{j^j e^{- ( 1 - e^{-
  t} ) j}}{j!} .
\end{equation}
and the limiting distribution is
\begin{equation}
  \label{eqn:K=x+yOptimalityLimDistFcn} F_* ( \hat{x} ) =
  \int_0^{\hat{x}} \frac{1}{\sqrt{2 \pi}} \hat{y}^{- 1 / 2} e^{- \hat{y} / 2}
  d \hat{y} .
\end{equation}

We consider now $\sup_{\hat{x} > 0} |F ( t, \hat{x} ) - F_* ( \hat{x} )
|$. Substituting Stirling's formula for the factorial in
(\ref{eqn:K=x+yOptimalityDistFcn}) and multiplying by an extra factor of $1 -
e^{- t}$, define
\[ \Phi ( t, \hat{x} ) = e^{- 2 t} \sum_{1 \leq j < e^{2 t} \hat{x}} ( 1 -
   e^{- t} )^j \frac{j^j e^{- ( 1 - e^{- t} ) j}}{\sqrt{2 \pi} j^{j +
   \frac{1}{2}} e^{- j}} . \]
In addition, define the Riemann sum corresponding to
(\ref{eqn:K=x+yOptimalityLimDistFcn}) by
\[ \Phi_* ( t, \hat{x} ) = e^{- 2 t} \sum_{1 \leq j < e^{2 t} \hat{x}}
   \frac{1}{\sqrt{2 \pi}} ( e^{- 2 t} j )^{- 1 / 2} e^{- ( e^{- 2 t} j ) / 2}
   . \]
Each of the differences
\begin{align}
  & D_1 = F ( t, \hat{x} ) - \Phi ( t, \hat{x} ) \nonumber\\
  & D_2 = \Phi ( t, x ) - \Phi_* ( t, \hat{x} ) 
  \label{eqn:K=x+yOptimalityDifferences}\\
  & D_3 = \Phi_* ( t, \hat{x} ) - F_* ( \hat{x} ) \nonumber
\end{align}
is estimated as follows.
For $D_1$, we can replace $F ( t, \hat{x} )$ by $( 1 - e^{- t} ) F ( t,
\hat{x} )$ since the difference between these terms is $O ( e^{- t} )$.
Utilizing the following error estimate for Stirling's approximation
{\cite{AS}}
\[ j! = \sqrt{2 \pi} j^{j + \frac{1}{2}} e^{- j} e^{\varepsilon ( j )},
   \qquad 0 < \varepsilon ( j ) < \frac{1}{12 j}, \]
we find that
\begin{align}
  D_1 & = e^{- 2 t} \sum_{1 \leq j < e^{2 t} \hat{x}} ( 1 - e^{- t} )^j
  \frac{j^j e^{- ( 1 - e^{- t} ) j}}{j!} \left( 1 - e^{\varepsilon ( j )}
  \right) + O ( e^{- t} ) \nonumber\\
  & = e^{- t} ( 1 - e^{- t} ) \mathbb{P} ( X_{1 - e^{- t}} < e^{2 t} \hat{x} ) + O
  ( e^{- t} ) . \nonumber
\end{align}
To estimate $D_3$, we use that the Riemann sum corresponding to an integral of
a differentiable function leads to an error proportional to the grid size.
Since $F'_* ( \hat{x} ) \in C^1$ we have that $D_3 = O ( e^{- 2 t} )$.
Finally, for the remaining term $D_2$ note that
\[ D_2 = e^{- 2 t} \sum_{1 \leq j < e^{2 t} \hat{x}} \frac{1}{\sqrt{2 \pi}} (
   e^{- 2 t} j )^{- 1 / 2} e^{- ( e^{- 2 t} j ) / 2} \left( e^{\alpha ( t, j
   )} - 1 \right) \]
where
\[ \alpha ( t, j ) = - j ( - e^{- t} - \log ( 1 - e^{- t} ) ) + e^{- 2 t} j /
   2 = - \frac{1}{3} e^{- 3 t} j + O ( e^{- 4 t} ) . \]
Therefore,
\[ D_2 = \frac{1}{3} e^{- t} \int_0^{\hat{x}} \frac{1}{\sqrt{2 \pi}}
   \hat{y}^{1 / 2} e^{- \hat{y} / 2} d \hat{y} + O ( e^{- 2 t} ) . \]
Putting together these estimates for (\ref{eqn:K=x+yOptimalityDifferences}) we
conclude that
\[ \sup_{\hat{x} > 0} |F ( t, \hat{x} ) - F_* ( \hat{x} ) | = O ( e^{- t}
   ) . \]

\section{Rate of convergence for the multiplicative kernel $K = x y$}
\label{sec:K=xy}

Lastly, we consider the case of the multiplicative kernel:

\begin{theorem}
  \label{thm:K=xy}Let $n_0$ be a positive measure such that $\int_{( 0, \infty )} x^2
  n_0 ( d x ) = \int_{( 0, \infty )} x^3 n_0 ( d x ) = 1$, and $\mu_4 =
  \int_{( 0, \infty )} x^4 n_0 ( d x ) < \infty$. With $\tau =
  \log ( 1 - t )^{- 1}$ and $\hat{n} ( t, d x )$ the rescaled solution to
  Smoluchowski's equation with $K = x y$ and initial data $n_0 ( d x )$, let
  \[ F ( \tau, \hat{x} ) = \int_{( 0, \hat{x} ]} \hat{y}^2 \hat{n} ( \tau, d
     \hat{y} ), \qquad F_* ( \hat{x} ) = \int_{( 0, \hat{x} ]} \hat{y}^2
     \hat{n}_{*, 2} ( \hat{y} ) d \hat{y} . \]
  Then for $\tau \in [ 0, \infty )$,
  \begin{equation}
    \label{eqn:K=xyMainResult} \sup_{\hat{x} > 0} |F ( \tau, \hat{x} ) -
    F_* ( \hat{x} ) | \leq C ( \mu_4 ) ( 1 + \tau ) e^{- \tau}
  \end{equation}
  where $C ( \mu_4 )$ is a constant that depends only on $\mu_4$.
\end{theorem}

\subsection{Proof for the multiplicative kernel} \label{subsec:K=xy,Proof}

The rate for the multiplicative kernel can be recovered from that of the
additive kernel by a classical change of variable due to Drake \cite{Drake}
(and discussed in \cite{MP2}). With initial data $n_0 ( d x )$, the
measure-valued solutions $n^{\text{mul}} ( t, d x )$ and $n^{\text{add}} ( t,
d x )$ to (\ref{eqn:Smoluchowski}) with $K = x y$ and $K = x + y$,
respectively, are related by
\[ x n^{\text{mul}} ( t, d x ) = ( 1 - t )^{- 1} n^{\text{add}} ( \tau ( t ),
   d x ), \qquad t \in ( 0, 1 ) . \]
Using the similarity variables
\[ \hat{x}_1 = \frac{x}{\lambda_1 ( t )}, \qquad \hat{x}_2 =
   \frac{x}{\lambda_2 ( t )} \]
given by (\ref{eqn:Scalings}) and the rescaled number distributions
(\ref{eqn:RescaledSoln}), we have for $\tau \in [ 0, \infty )$ that
\[ \hat{x}_2^2 \hat{n}^{\text{mul}} ( \tau, d \hat{x}_2 ) = \hat{x}_1
   \hat{n}^{\text{add}} ( \tau, d \hat{x}_1 ) . \]
Therefore,
\[ F_2 ( \tau, \hat{x} ) = \int_0^{\hat{x}} \hat{y}_2^2 \hat{n}^{\text{mul}} (
   \tau, d \hat{y}_2 ) = \int_0^{\hat{x}} \hat{y}_1 \hat{n}^{\text{add}} (
   \tau, d \hat{y}_1 ) = F_1 ( \tau, \hat{x} ) . \]
Since (\ref{eqn:SelfSimilarProfiles}) implies that the limiting distributions
satisfy $F_{*, 2} ( \hat{x} ) = F_{*, 1} ( \hat{x} )$, we combine this
with the previous equation to get (\ref{eqn:K=xyMainResult}):
\[ \sup_{\hat{x} > 0} |F_2 ( \tau, \hat{x} ) - F_{*, 2} ( \hat{x} ) | =
   \sup_{\hat{x} > 0} |F_1 ( \tau, \hat{x} ) - F_{*, 1} ( \hat{x} ) | \leq
   C ( \mu_4 ) ( 1 + \tau ) e^{- \tau} . \]
The constant $C$ is the same as that obtained in the case $K = x + y$,
with $\mu_3$ replaced by $\mu_4$. Lastly, the convergence rate is nearly
optimal for monodisperse initial data by making the change of variables to the
additive kernel case.

\section*{Acknowledgements}
The author gratefully acknowledges Govind Menon for his guidance and support, and Bob Pego
for many useful discussions. An anonymous referee is also thanked for comments that
greatly improved the clarity of this article. This work is supported by the National Science
Foundation under grants DMS 06-05006 and DMS 07-48482.

\bibliographystyle{siam}
\bibliography{s1}
\end{document}